\documentclass{article}
\usepackage[nonatbib]{arxiv}
\usepackage[utf8]{inputenc} 
\usepackage[T1]{fontenc}    
\usepackage{hyperref}       
\usepackage{url}            
\usepackage{booktabs}       
\usepackage{amsfonts}       
\usepackage{nicefrac}       
\usepackage{microtype}      
\usepackage{graphicx} 
\usepackage{caption}
\usepackage{amsmath}
\usepackage{algorithm2e}
\usepackage{caption,float}
\usepackage{enumitem}
\usepackage{subcaption}
\usepackage{xcolor}
\usepackage{soul}
\usepackage{amsthm}

\sethlcolor{red}

\newcommand{\diff}{\mathop{}\mathopen{}\mathrm{d}}
\newtheorem{theorem}{Theorem}
\newtheorem{corollary}{Corollary}

\newcommand{\approxpropto}{\mathrel{\vcenter{
  \offinterlineskip\halign{\hfil$##$\cr
    \propto\cr\noalign{\kern2pt}\sim\cr\noalign{\kern-2pt}}}}}

\newcommand*\samethanks[1][\value{footnote}]{\footnotemark[#1]}

\title{Ghost Units Yield Biologically Plausible Backprop in Deep Neural Networks}
 
\author{{\large \bf Thomas Mesnard\thanks{Both authors contributed equally to this paper. Alphabetical order.}~~(thomas.mesnard@gmail.com)} \\
  Mila - Quebec AI Institute\\
  Universite de Montreal, Montreal, Canada
  \AND {\large \bf Gaëtan Vignoud\samethanks~~(gaetan.vignoud@gmail.com)} \\
  Mila - Quebec AI Institute\\
  Universite de Montreal, Montreal, Canada
  \AND {\large \bf João Sacramento (sacramento@ini.ethz.ch)} \\
  Institute of Neuroinformatics\\
  University of Zurich and ETH Zurich, Switzerland
  \AND {\large \bf Walter Senn (senn@pyl.unibe.ch)} \\
  Department of Physiology, University of Bern\\ Bern 3012, Switzerland\\
  \AND {\large \bf Yoshua Bengio\thanks{CIFAR Fellow}~  (yoshua.bengio@mila.quebec)} \\
  Mila - Quebec AI Institute\\
  Universite de Montreal, Montreal, Canada}

\begin{document}
\maketitle

\begin{abstract}
In the past few years, deep learning has transformed artificial intelligence research and led to impressive performance in various difficult tasks. However, it is still unclear how the brain can perform credit assignment across many areas as efficiently as backpropagation does in deep neural networks.
In this paper, we introduce a model that relies on a new role for a neuronal inhibitory machinery,  referred to as ghost units. By cancelling the feedback coming from the upper layer when no target signal is provided to the top layer, the ghost units enables the network to backpropagate errors and do efficient credit assignment in deep structures. While considering one-compartment neurons and requiring very few biological assumptions, it is able to approximate the error gradient and achieve good performance on classification tasks. Error backpropagation occurs through the recurrent dynamics of the network and thanks to biologically plausible local learning rules. In particular, it does not require separate feedforward and feedback circuits. Different mechanisms for cancelling the feedback were studied, ranging from complete duplication of the connectivity by long term processes to online replication of the feedback activity. This reduced system combines the essential elements to have a working biologically abstracted analogue of backpropagation with a simple formulation and proofs of the associated results. Therefore, this model is a step towards understanding how learning and memory are implemented in cortical multilayer structures, but it also raises interesting perspectives for neuromorphic hardware.
\end{abstract}


\section{Introduction}
\label{sec:intro}

Recently, deep learning~\cite{Goodfellow-et-al-2016-Book} has revolutionized artificial intelligence and led to impressive performance in various tasks such as computer vision~\cite{krizhevsky2012imagenet}, speech recognition~\cite{hinton2012deep} and machine translation~\cite{bahdanau2014neural}. Deep learning, thanks to backpropagation~\cite{almeida1987learning, pineda1987generalization}, is able to take advantage of the multilayer structure of the neural network to learn high level features relevant for the given task it is trained to perform. Those features are more and more abstract while going deeper in the network, and give rise to a high level representation of the data. In the case of visual tasks, when using convolutional neural networks~\cite{lecun1990handwritten}, the high level features learned thanks to the backpropagation algorithm are similar to the ones experimentally observed in the visual cortex~\cite{khaligh2014deep}.\\

However, it is still an open question how the brain is able to perform credit assignment in deep neural structures spanning multiple areas. The core algorithm used to train deep neural networks (i.e backpropagation) has been seen by the neuroscience community as being biologically implausible because the implementation used in deep learning relies on assumptions that cannot be met in the brain~\cite{bengio2015towards, neftci2017event} (i.e need for symmetric weights and a separate circuit for feedforward and gradient computations, precise timing between the forward and the backward paths with fixed activity of the neurons and knowledge of the derivative of the forward activation to correctly update the weights, high precision numbers to characterize forward activities and backpropagate errors compared to binary values occurring in the brain). Thanks to recent work~\cite{lillicrap2016random, nokland2016direct} the assumption that the feedback weights must be the exact transpose of the forward ones is no longer required to have efficient credit assignment thanks to the feedback-alignment mechanism. \cite{courbariaux2016binarized} showed it was possible to train deep neural networks using backpropagation with binarized activation functions and binary weights. \cite{scellier2017equilibrium} showed how the same neurons could be used for feedforward and gradient computations thanks to the recurrence induced by feedback connections and how nudging output units towards a lower-error configuration propagates, via feedback connections, error gradients in the inner layers of the circuit.\\

Recent studies have shown that local contrastive hebbian plasticity in an energy based model can implement backpropagation in deep neural structures thanks to the recurrent dynamics~\cite{bengio2017stdp, scellier2017equilibrium} while having promising results when using leaky integrate-and-fire neurons~\cite{mesnard2016towards}. Others~\cite{guerguiev2017towards} were able to train deep learning networks by approximating backpropagation in systems with multicompartment neurons.
Finally, a recent study~\cite{sacramento2018dendritic,NIPS2018_8089} used recurrent networks with inhibitory neurons to approximate backpropagation. Those inhibitory units aim to predict and cancel the feedback signal coming from the upper layers. When they are not able to correctly predict the incoming feedback, weights are updated proportionally to this prediction error which is closely related to the correct gradient that would have been obtained with classical backpropagation. This idea of predicting the upper incoming feedback activity to predict the correct gradient can be related to~\cite{jaderberg2016decoupled}, where side networks are introduced between each layer of the main network that learn in a supervised manner to predict the correct gradients based only on the feedforward inputs. A more thorough comparison with previous works on how backpropagation could be implemented in the brain is done in Section~\ref{sec:related_work}.
\\

In this paper, we consider a recurrent network composed of pyramidal units (PU) that can be identified with the feedforward units of a multilayer perceptron, with a dynamic of learning inspired from~\cite{scellier2017equilibrium} with two different phases. These cells integrate feedforward activity coming from the lower layers but also feedback activity coming from the upper layers. Moreover, in order to enable backpropagation of errors, we introduce a new type of interneurons, referred to as ghost units (GU). Their goal is to predict and cancel feedback from pyramidal units in the upper layer by integrating the same feedforward input without having access to any feedback from the following layers.
This property enables the network to converge quickly during the feedforward computation, in spite of the presence of recurrent connections, by canceling feedback coming from the upper layers thanks to the ghost units. This cancellation effect of the ghost units allows top-down corrective feedback to be correctly backpropagated when targets are provided in the weakly-clamped phase.
This gives the network the capacity to perform credit assignment in a multilayer structure by simply following its dynamics and updating the weights according to local plasticity learning rules.

\section{Backpropagation thanks to ghost units in a recurrent and dynamical neural network}
\label{sec:backprop}

\subsection{Architecture}
\label{subsec:arch}

We consider a biologically plausible implementation for backpropagation in a directed acyclic graph of feedforward connections with network input $x$. We consider that the network has $k+1$ layers. We will use $l=0$ for the first layer that represents the inputs and $l=k$ for the last layer which is the output of the network, see Figure~\ref{fig:ma_architecture} for a schematic representation of the architecture.

At each layer $l$, a node of this graph is associated with one or multiple pyramidal units (PU)\footnote{Note that a single feedforward unit (called pyramidal unit) could in reality be implemented by multiple pyramidal units with similar input and output connectivity, allowing the network to reduce the spiking noise, if integrate-and-fire neurons where used for example.} whose activity is denoted by $s_l[i]$ (the state of unit $i$ in layer $l$). Pyramidal units have an output non-linearity $\rho$ which maps their activity $s_l[i]$ to their firing rate $\rho(s_l[i])$. Technically, this transfer function must be $K$-Lipschitz continuous. $\mathcal{S}_l$ is the set of pyramidal units in layer $l$.

For a connectivity matrix $A$, we define $A_{\text{output layer},\text{input layer}}[i,j]$ which corresponds to the synaptic weight from unit $j$ in the input layer to unit $i$ in the output layer.

Both feedforward and feedback connections are considered in this model. The main feedforward synaptic weights $W^f_{l+1,l}[i,j]$ correspond to the influence of presynaptic unit $j$ (in layer $l$) on the postsynaptic unit $i$ (in layer $l+1$). The feedback weights $W^b_{l,l+1}[i,j]$ encapsulate the effect of the pyramidal unit $j$ (of layer $l+1$) on pyramidal unit $i$ (layer $l$). The network output is defined by the firing rate $\rho(s_k)$ of the output pyramidal units. This output is compared to target values $t_k$ and the performance is measured through a scalar cost function $C(\rho(s_k),t_k)$ which somehow compares the network output and the target. In the simulations, the mean-squared error function was used as cost function: $C(\rho(s_k),t_k)=\sum_{j \in \mathcal{S}_k}(\rho(s_k[j])-t_k[j])^2$, but other losses could also be implemented in the same way.\\

We also define the cost function $\tilde{C}(\rho(\tilde{s}_k),t_k)=\sum_{j \in \mathcal{S}_k}(\rho(\tilde{s}_k[j])-t_k[j])^2$ which corresponds to the cost function of the associated (purely feedforward) multilayer perceptron (MLP), where the associated activity is defined by:
\begin{equation}
    \tilde{s}_l = W^f_{l,l-1} \rho(\tilde{s}_{l-1})
\end{equation}

Training is decomposed in a free phase and a weakly-clamped phase following~\cite{scellier2017equilibrium}. During the free phase, the network evolves thanks to its recurrent dynamics with only inputs provided. During the weakly-clamped phase, both inputs and targets are presented to the network. A top-down error signal $-\beta\frac{\partial  C(\rho(s_k),t_k)}{\partial  \rho(s_k)}$ pushes the output units $s_k$ towards a value corresponding to a smaller loss $C$. $\beta=0$ corresponds to the free-phase and $\beta>0$ to the weakly-clamped one.

In addition, we consider a lateral network of ghost units (GU), which could be implemented by inhibitory interneurons. These units are represented by a scalar variable $g_l[i]$ for each unit $i$ in layer $l$. A ghost unit in a layer is only connected to the pyramidal units of the same layer, through two matrices $V^f_{l,l}[i,j]$ (for  lateral connections from the pyramidal unit $j$ to ghost unit $i$) and $V^b_{l,l}[i,j]$ (for the lateral connections from the ghost unit $j$ to the pyramidal unit $i$). These units aim to reproduce the feedback activity from the pyramidal units of the next layer during the forward phase, and therefore enable the network to directly compute the gradient during the weakly-clamped phase. These units are considered as inhibitory when projecting to the pyramidal neurons (expressed here as minus sign in $-V^b_{l,l}\rho(g_l)$, although the synaptic weights $V^b_{l,l}$ can themselves be negative). These ghost units are only present at each hidden layer $l$ ($l\neq 0$ and $l\neq k$).\\

We will show in the following section that the combination of lateral recurrent and feedback connections propagates the error through the network in a way that closely approximates backpropagation, so long as some assumptions are satisfied, regarding the ability of feedback connections to mimic feedforward connections (approximate symmetry) and of lateral connections to learn to cancel the feedback connections when there is no nudging.\\

\subsection{Notations}
\label{subsec:notations}
\begin{center}
\begin{tabular}{cccc}
   $x$ & network input & $\tau $& time constant \\
     PU & pyramidal unit & GU & ghost unit \\
   $s_l[i]$ & activity of PU $i$ in layer $l$ &  $C(\rho(s_k),t_k)$ & cost function \\
      $\tilde{s}_l[i]$ & activity of PU $i$ in layer $l$ in the MLP &  $\tilde{C}(\rho(\tilde{s}_k),t_k)$ & cost function of the MLP \\
   $\rho$ & neuronal transfer function & $g_l[i]$ & activity of GU  $i$ in layer $l$ \\
   $\mathcal{S}_l$ & \multicolumn{3}{c}{set of pyramidal units in layer $l$} \\
   $W^f_{l+1,l}[i,j]$ & \multicolumn{3}{c}{feedforward connection from PU $j$ (layer $l$) to PU $i$ (layer $l+1$)} \\
   $W^b_{l,l+1}[i,j]$ & \multicolumn{3}{c}{feedback connection from PU $j$ (layer $l+1$) to PU $i$ (layer $l$)}\\
   $V^f_{l,l}[i,j]$ & \multicolumn{3}{c}{lateral (recurrent) connection from PU $j$ to GU $i$ (layer $l$)}\\
   $V^b_{l,l}[i,j]$ & \multicolumn{3}{c}{lateral (recurrent) connection from GU $j$ back to PU $i$ (layer $l$)}
\end{tabular}
\end{center}

\subsection{Dynamics of the neurons}
\label{subsec:dynamics}

Three different inputs are integrated by pyramidal units in layer $l$:
\begin{itemize}[noitemsep,topsep=0pt]
\item $b_l = W^f_{l,l-1} \rho(s_{l-1})$ is the bottom-up input coming from the pyramidal units of layer $l-1$.
\item $t_l = W^b_{l,l+1} \rho(s_{l+1})$ is the top-down feedback coming from pyramidal units of layer $l+1$.
\item $c_l = V^b_{l,l} \rho(g_l)$ is the lateral feedback coming from the ghost units of the same layer $l$.
\end{itemize}

The pyramidal units are evolving through:
\begin{align}
  \tau\dot{s}_l = - s_l + b_l + e_l
\label{eq:dynamics_s}
\end{align}

where $e_l$ represents an error term whose expression depends on the layer. 

For the hidden layers, $e_l=t_l - c_l$ is the difference between the top-down feedback (the local target $t_l$) and the cancelling contribution from the inhibitory ghost units ($c_l$, counted negatively because of the inhibitory nature of the ghost units). For the output layer, $e_k= -\beta \frac{\partial  C}{\partial  \rho(s_k)}$ is the nudging term that indicates in which direction $s_k$ should move to reduce the output cost function $C=\sum_{j \in \mathcal{S}_k}(\rho(s_k[j])-t_k[j])^2$, with $t_k$ the target output values ($\beta=0$ in the free phase and $\beta>0$ in the weakly-clamped phase).

In particular, at the end of the forward phase when $e_l\rightarrow 0$, we have: $s_l=\tilde{s}_l=W^f_{l,l-1} \rho(\tilde{s}_{l-1})$. Because of perfect cancellation $t_l=c_l$, the network behaves like a feedforward multilayer perceptron.

The ghost units of layer $l$ follow:
\begin{align}
  \tau\dot{g}_l = -g_l + V^f_{l,l}\rho(s_l)
\label{eq:dynamics_g}
\end{align}

\section{Different architectures and learning procedures}
\label{sec:models}

\subsection{Network with 1-1 correspondence between the pyramidal units and the ghost units (MA)}
\label{subsec:ma}

\paragraph{Model description}
\label{subsubsec:ma_desc}

In this section, we consider that each pyramidal unit $j$ in layer $l+1$ has a corresponding ghost unit\footnote{In biology there are more pyramidal neurons than inhibitory neurons. Yet, GU may also include sub-classes of pyramidal neurons, such that the number of GU must not be smaller than the number of PU. Moreover, the code formed by pyramidal neurons may show some redundancy so that it could be compressed to a smaller number of effective PU.} in the previous layer $l$, and that the ghost units aim to replicate the activity of their associated pyramidal units by integrating the same inputs. In order to make the reading easier in this part, we use the same indices in the brackets for the ghost unit and its associated pyramidal unit. For example, ghost unit $j$ of layer $l$ (with activity $g_l[j]$) will be associated to the pyramidal unit $j$ of layer $l+1$ (activity $s_{l+1}[j]$). This architecture can be seen in Figure~\ref{fig:ma_architecture}.

\begin{figure}[ht]
  \centering
  \includegraphics[width=0.6\linewidth]{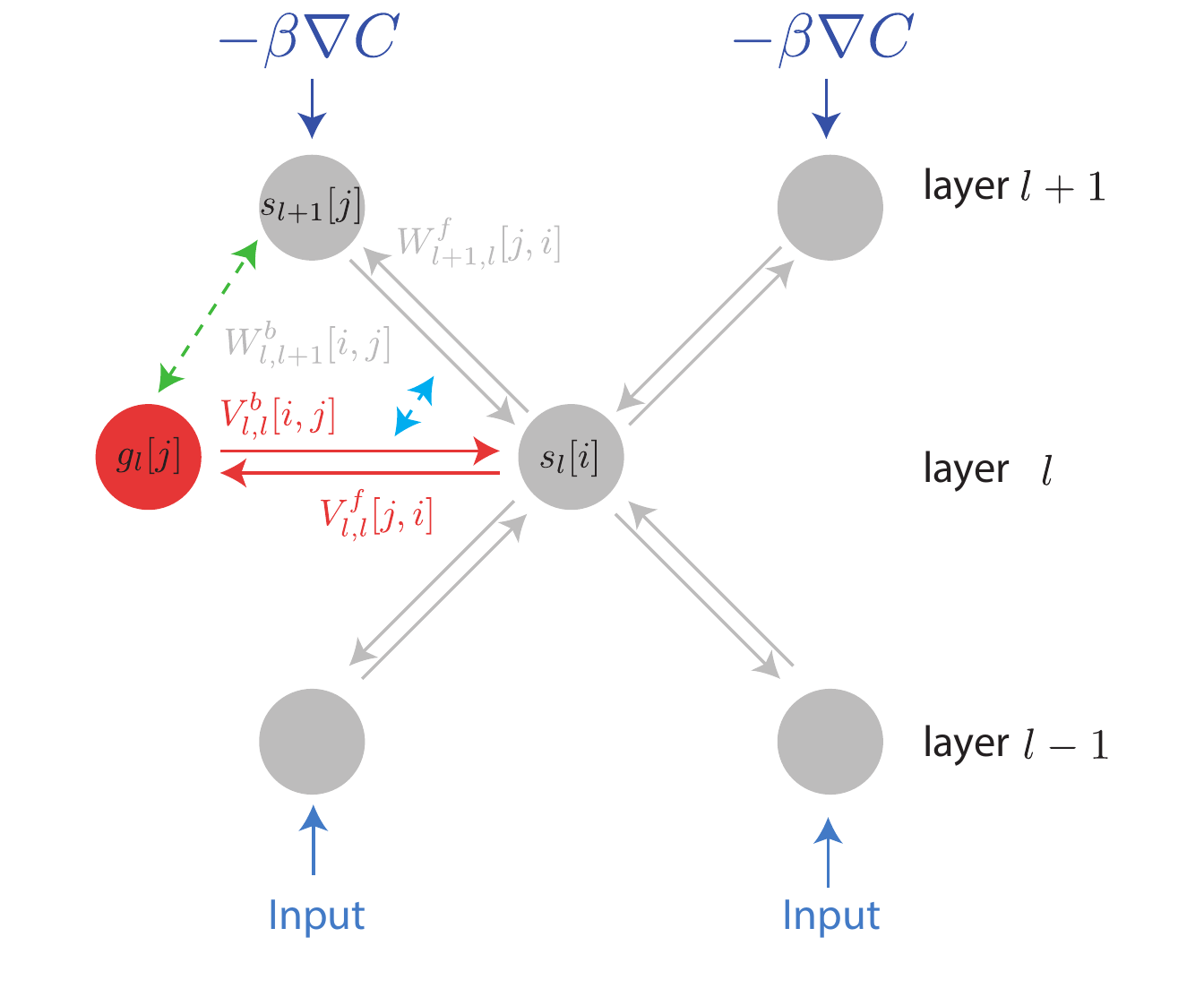}
\caption{Architecture for Model A (MA) network with 1-1 correspondence between ghost units and associated pyramidal units from the following layer. Pyramidal units (in grey) and ghost units (in orange) are connected through different weights matrices ($W^f$, $W^b$ in grey for PU-PU connectivity and $V^f$, $V^b$ in orange for PU-GU lateral connections). Ghost units of layer $l$ tend to copy the pyramidal unit activity of the following layer $l+1$, thanks to the evolution of $V^f$ eq.~\ref{eq:ma_dvf} (this interaction is represented by the green dotted line) and at the same time, to cancel the feedback coming from the pyramidal units of the next layer, because of $V^b$ updates eq.~\ref{eq:ma_dvb} (blue dotted line). The nudging (-$\beta \nabla C$) is presented at the output layer ($\beta=0$ in the free phase and $\beta \neq 0$ in the weakly-clamped phase).
}
\label{fig:ma_architecture}
\end{figure}

During the free-phase, only the lateral connections between the pyramidal and ghost units are updated. The local learning rules for the synaptic weights $V^f$ and $V^b$ are defined as follow:

\begin{itemize}[noitemsep,topsep=0pt]
\item $s_{l+1}$ acts like a target for the ghost units $g_l$ to learn $V^f_{l,l}$:
  \begin{align}
    \dot{V}^f_{l,l} = \eta_{V}(s_{l+1} - g_l) \rho(s_l)^T
    \label{eq:ma_dvf}
  \end{align}
  This minimizes $\|s_{l+1} - g_{l}\|$, i.e., the inhibitory ghost unit learns to imitate its associated pyramidal unit. $\eta_V$ is the learning rate.
  
\item  The top-down feedback $t_l$ onto layer $l$ acts as a target for the weights forming
  the canceling feedback $c_l$:
  \begin{align}
    \dot{V}^b_{l,l} = \eta_{V}(t_l - c_l) \rho(g_l)^T
    \label{eq:ma_dvb}
  \end{align}
  This minimizes $\|t_l - c_l\|$ for each layer $l$, with the same learning rate $\eta_V$.
  \end{itemize}
  During the weakly-clamped phase, only the $W^f$ and $W^b$ are updated through the following learning rules:  
  \begin{itemize}
\item The main weights (feedforward, from pyramidal units of layer $l$ to pyramidal units of layer $l+1$) are updated, using a local learning rule:
  \begin{align}
  \label{eq:ma_dw}
    \dot{W}^f_{l+1,l} =\eta_{W} \left(e_{l+1}\cdot \rho'(s_{l+1})\right) \rho(s_l)^T
  \end{align}
  This approximates gradient descent on $\tilde{C}$, see Theorem~\ref{thm:ma_bp}.
  \item The feedback weights are set equal to the transpose of the feedforward ones: $W^b_{l,l+1}=(W^f_{l+1,l})^T$.
\end{itemize}

This was implemented using Euler discretization, see Algorithm~\ref{algo:ma} for a more precise description of the algorithm.

We also consider a variant (MA') where all the updates are performed during all the phases. This version, with continuous updates of the weights is more close to what can be expected to happen in the brain.

\paragraph{A good approximation of backpropagation}
\label{subsubsec:ma_bp} The combination of these update rules, leads to the following theorem, where eq.~\ref{eq:ma_dw} can be seen as a close estimate to backpropagation.

\begin{theorem}
\label{thm:ma_bp}  
   We set the backward weights equal to the (adapting) forward weights,  $W^b_{l,l+1}=(W^f_{l+1,l})^T$, and assume that the ghost unit circuit ($V^f$, $V^b$) converged during the free phases, $\|s_{l+1} - g_l\|\rightarrow 0$, $\|t_l - c_l\|\rightarrow 0$ for all hidden layer $l$ (MA). Then for weak output nudging ($\beta$ small, $e_k= -\beta \frac{\partial  C}{\partial \rho(s_k)}$) errors converge to $e_l \approx -\beta \frac{\diff  \tilde{C}}{\diff  \rho(\tilde{s}_l)}$ for each hidden layer~$l$. So, the forward weights ($W^f$) become updated according to the classical backpropagated error gradient. 
\end{theorem}  

\begin{proof}

Let's introduce $\|A\|$ which is the L2 norm if $A$ is a vector and the maximum singular value norm if $A$ is a matrix.

We suppose that all the matrices are bounded during the procedure.
In other words, we suppose that there exists $M$ such that $\|V^f\|<M$, $\|W^f\|<M$, $\|V^b\|<M$ and  $\|W^b\|<M$. This could be ensured by clipping the weights between extremal values.

We remind the definition of $\tilde{C}$, that corresponds to the cost function of the multilayer perceptron built from the feedforward graph (with no recurrent dynamics) and nodes $\tilde{s}_l$.

\emph{Free phase:}

Firstly, we study the dynamics of the weights during the free phase, where both loss functions $\|t_l - c_l\|$ and $\|s_{l+1} - g_l\|$ are minimized thanks to the updates eq.~\ref{eq:ma_dvf} and eq.~\ref{eq:ma_dvb}.\\

In the stationary limit (where $\dot{s}_l=\dot{g}_l = 0$), we have $s_l=b_l+t_l-c_l$ and so, for all hidden layer $l$:

  \begin{align*}
    \|W^f_{l+1,l}\rho(s_l) - V^f_{l,l}\rho(s_l)\| & = \|s_{l+1} - g_l - t_{l+1} + c_{l+1}\|\\
     & \leq \|s_{l+1} - g_l\|+\|t_{l+1} - c_{l+1}\|
  \end{align*}

So if both $\|t_{l+1} - c_{l+1}\|$ and $\|s_{l+1} - g_l\|$ are minimized for all $l$ in the stationary limit, then $\|W^f_{l+1,l}\rho(s_l) - V^f_{l,l}\rho(s_l)\|$ tends to $0$. Considering that the space spanned by $\rho(s_l)$ during learning is large enough (this means having a large set of training data, which was true in the cases we tested), the mean-squared error of the associated linear regression converges also to $0$, and therefore $V^f_{l,l}\rightarrow W^f_{l+1,l}$ during the free phase.

Secondly, still in the stationary limit,
  \begin{align*}
    \|W^b_{l,l+1}\rho(s_{l+1}) - V^b_{l,l}\rho(s_{l+1})\| & = \|t_l - c_l + V^b_{l,l}(\rho(g_l)-\rho(s_{l+1}))\|\\
     & \leq \|t_l - c_l\|+M\|\rho(g_l)-\rho(s_{l+1})\|\\
     & \leq \|t_l - c_l\|+MK\|s_{l+1} - g_l\|
  \end{align*}
  
because $\rho$ is $K$-Lipschitz.
So if both $\|t_{l+1} - c_{l+1}\|$ and $\|s_{l+1} - g_l\|$ are minimized for all $l$, then $ \|W^b_{l,l+1}\rho(s_{l+1}) - V^b_{l,l}\rho(s_{l+1})\|\rightarrow 0$. As before, if we consider that the space spanned by $\rho(s_{l+1})$ is big enough, we also minimize the mean-squared error of the linear regression, and thus $V^b_{l,l}\rightarrow W^b_{l,l+1}=(W^f_{l+1,l})^T$ during the free phase.\\

\emph{Weakly-clamped phase:}

For the weakly-clamped phase, we prove the theorem by induction over the layers. We suppose that the learning of the free phase is done and so that $V^f_{l,l} = (V^b_{l,l})^T = W^f_{l+1,l}$. 

Firstly, it is easy to see that for the output layer (index $k$), in the stationary limit and with small nudging, we have $s_k=\tilde{s}_k$ (at zero order in $\beta$) and so:

\begin{equation}
    e_k= -\beta \frac{\partial  C}{\partial  \rho(s_k)} = -\beta \frac{\partial  \tilde{C}}{\partial  \rho(\tilde{s}_k)}
\end{equation}

Then we just have to prove that this property is true for the last hidden layer $k-1$ and the rest of the proof will follow by induction.

Considering that the layer $k-1$ is still at equilibrium, and that the nudging was small, we have $s_{k-1}=\tilde{s}_{k-1}$ and so in the stationary limit:
  
   \begin{align*}
    s_{k}&=W^f_{k,k-1}\rho(\tilde{s}_{k-1})+e_k\\
       &=\tilde{s}_k+e_k
  \end{align*}
  Contrarily to just above, we need the first order approximation for $s_k$ in $\beta$ here (otherwise we would get $e_{k-1}=0$ in the following).
  
  As we have $V^f_{k-1,k-1}=W^f_{k,k-1}$ and $s_{k-1}=\tilde{s}_{k-1}$, 
    \begin{align*}
    g_{k-1}&=V^f_{k-1,k-1}\rho(\tilde{s}_{k-1})\\
    &=W^f_{k,k-1}\rho(\tilde{s}_{k-1})\\
       &=\tilde{s}_k
  \end{align*} 
  
 Starting from the definition of $e_{k-1}$, we substitute the definition of $t_{k-1}$ and $c_{k-1}$. Using that $V^b_{k-1,k-1}=W^b_{k-1,k}$ (from free phase) and $g_{k-1}=\tilde{s}_k$, we have:
  \begin{align*}
    e_{k-1} &= t_{k-1} - c_{k-1} \\
        &= W^b_{k-1,k} \rho(s_k) - V^b_{k-1,k-1} \rho(g_{k-1})\\
        &= W^b_{k-1,k} \rho(s_k) - W^b_{k-1,k} \rho(g_{k-1})\\
        &= W^b_{k-1,k}(\rho(\tilde{s}_k+e_k) -  \rho(\tilde{s}_k))\\
\end{align*}
Then, by assuming $e_k$ is small (because $\beta$ is small):
\begin{align}
\label{eq:ei}
     e_{k-1} &=W^b_{k-1,k}(\rho'(\tilde{s}_k)\cdot e_k + o(e_k))\\
     e_{k-1}[i]  &\approx \sum_{j \in \mathcal{S}_k} W^b_{k-1,k}[i,j] \rho'(\tilde{s}_k[j])e_k[j]
  \end{align}
  
By the chain rule on the feedforward graph, we have:
\begin{equation}
\label{eq:chainrule}
    \sum_{j \in \mathcal{S}_k} W^f_{k,k-1}[j,i] \rho'(\tilde{s}_k[j]) \frac{\diff  \tilde{C}}{\diff  \rho(\tilde{s}_k[j])} = \frac{\diff  \tilde{C}}{\diff  \rho(\tilde{s}_{k-1}[i])}
\end{equation}

From eq.~\ref{eq:ei}, eq.~\ref{eq:chainrule} and $W^b_{k-1,k}=(W^f_{k,k-1})^T$:
  \begin{align}
    e_{k-1} \approx  -\beta\frac{\diff  \tilde{C}}{\diff  \rho(\tilde{s}_{k-1})}
  \end{align}

Using induction across layers, we have for every layer $l$:
  \begin{align}
    e_l \approx -\beta \frac{\diff  \tilde{C}}{\diff  \rho(\tilde{s}_l)}.
  \end{align}
\end{proof}  
Due to the stacked nonlinearities the approximation for deeper layers may get worse and worse as we go deeper, this can be compensated by choosing smaller $\beta$.

\begin{minipage}{0.5\textwidth}
\begin{algorithm}[H]
 \textbf{initialization}\\
 \While{not done}{
 Sample batch from the training set \\
  \For{\_ in range free\_steps}{
   for output units $e_k=0$\\
   for hidden units $e_l=W_{l,l+1}^b\rho(s_{l+1})-V_{l,l}^b\rho(g_l)$\\
   $\forall$ layer $l$,\\
   $s_l \leftarrow s_l+\frac{\diff t}{\tau}[-s_l+W_{l,l-1}^f\rho(s_{l-1})+e_l]$\\
   $g_l \leftarrow g_l+\frac{\diff t}{\tau}[-g_l+V_{l,l}^f\rho(s_l)]$\\
   $V_{l,l}^f \leftarrow V^f_{l,l}+\eta_V\diff t[(s_{l+1} - g_l) \rho(s_l)^T]$\\
    $V_{l,l}^b \leftarrow V^b_{l,l}+\eta_V\diff t [e_l \rho(g_l)^T]$\\
   }
  \For{\_ in range weakly\_clamped\_steps}{
   for output units $e_k=-\beta\frac{\diff  C}{\diff  \rho(s_k)}$\\
   for hidden units $e_l=W_{l,l+1}^b\rho(s_{l+1})-V_{l,l}^b\rho(g_l)$\\
   $\forall$ layer $l$, \\
   $s_l \leftarrow s_l+\frac{\diff t}{\tau}[-s_l+W_{l,l-1}^f\rho(s_{l-1})+e_l]$\\
   $g_l \leftarrow g_l+\frac{\diff t}{\tau}[-g_l+V_{l,l}^f\rho(s_l)]$\\
   $W_{l,l-1}^f \leftarrow W_{l,l-1}^f+\eta_W\diff t \left(e_{l}\cdot \rho'(s_{l})\right) \rho(s_{l-1})^T$\\
    $W_{l-1,l}^b \leftarrow (W_{l,l-1}^f)^T$\\
   }   
 }
 \captionof{algocf}{Learning procedure for Model A \newline (MA) network.}
 \label{algo:ma}
\end{algorithm}
\end{minipage}
\begin{minipage}{0.5\textwidth}
\begin{algorithm}[H]
\textbf{initialization}\\
 \While{not done}{
 Sample batch from the training set \\
  \For{\_ in range free\_steps}{
   for output units $e_k=0$\\
   for hidden units $e_l=W_{l,l+1}^b\rho(s_{l+1})-V_{l,l}^b\rho(g_l)$\\
   $\forall$ layer $l$,\\
   $s_l \leftarrow s_l+\frac{\diff t}{\tau}[-s_l+W_{l,l-1}^f\rho(s_{l-1})+e_l]$\\
   $g_l \leftarrow g_l+\frac{\diff t}{\tau}[-g_l+V_{l,l}^f\rho(s_l)]$\\
   $V_{l,l}^b \leftarrow V^b_{l,l}+\eta_V\diff t [e_l \rho(g_l)^T]$\\
   }
  \For{\_ in range weakly\_clamped\_steps}{
   for output units $e_k=-\beta\frac{\diff  C}{\diff  \rho(s_k)}$\\
   for hidden units $e_l=W_{l,l+1}^b\rho(s_{l+1})-V_{l,l}^b\rho(g_l)$\\
   $\forall$ layer $l$,\\
   $s_l \leftarrow s_l+\frac{\diff t}{\tau}[-s_l+W_{l,l-1}^f\rho(s_{l-1})+e_l]$\\
   $g_l \leftarrow g_l+\frac{\diff t}{\tau}[-g_l+V_{l,l}^f\rho(s_l)]$\\
   }   
   {
    $W_{l,l-1}^f \leftarrow W_{l,l-1}^f+\eta_{W,l} \left(e_{l}\cdot \rho'(s_{l})\right) \rho(s_{l-1})^T$\\
    $W_{l-1,l}^b \leftarrow (W_{l,l-1}^f)^T$\\
  }
 }
 \captionof{algocf}{Learning procedure for model B \newline (MB) network.}
 \label{algo:mb}
\end{algorithm}
\end{minipage}

\begin{corollary}
\label{cor:ma_bp}
  Under the assumptions of Theorem~\ref{thm:ma_bp}, the weight change proposed in eq.~\ref{eq:ma_dw} corresponds to approximate stochastic gradient descent, i.e.,
  \begin{align}
    \dot{W}^f_{k-1,k} \approx -\eta_W\beta\frac{\diff  \tilde{C}}{\diff  W^f_{k,k-1}}
  \end{align}
\end{corollary}

\begin{proof}
  \begin{align}
    \frac{\diff  \tilde{C}}{\diff  W^f_{l,l-1}[i,j]} &=
    \frac{\diff  \tilde{C}}{\diff  \rho(\tilde{s}_l[i])}\frac{\diff  \rho(\tilde{s}_l[i])}{\diff  \tilde{s}_l[i]}\frac{\diff  \tilde{s}_l[i]}{\diff  W^f_{l,l-1}[i,j]}
    \nonumber \\
    &= \frac{\diff  \tilde{C}}{\diff  \rho(\tilde{s}_l[i])} \rho'(\tilde{s}_l[i]) \rho(\tilde{s}_{l-1}[j])
  \end{align}
  Hence, if $e_l \approx -\beta \frac{\diff  \tilde{C}}{\diff  \rho(\tilde{s}_l)}$,
  we obtain that $e_l[i] \rho'(s_l[i]) \rho(s_{l-1}[j]) \approx -\beta\frac{\diff  \tilde{C}}{\diff  W^f_{l,l-1}[i,j]}$ and the corollary.
\end{proof}

\subsection{Deep neural network with ghost units replicating online the feedback from the pyramidal units (MB)}
\label{subsec:mb}
\paragraph{Model description}
\label{subsubsec:mb_desc}

We also developed a different class of models that we introduce in this section.  In this model (MB), we do not make any hypothesis on the number of pyramidal units and ghost units. We also consider that the lateral connections $V^f_{l,l}$ from the pyramidal units of layer $l$ to the ghost units of the same layer are fixed to a randomly initialized value, therefore $\dot{V}^f_{l,l} = 0$. $V^b$ evolves so as to have, example by example, the feedback $c_l$ coming from the ghost units of layer $l$ replicating the feedback $t_l$ coming from the pyramidal units of layer $l+1$, see Figure~\ref{fig:mb_architecture}. Thanks to this property, we have $c_l=t_l$ for a given example at the end of the free phase after the efficient and rapid learning of $V^b$. This enables the network to correctly learn $W^f$ in the weakly-clamped phase.\\
This highly modular and fast changing plasticity could be implemented in real neural circuits by Post-Tetanic Potentiation. This type of plasticity evolves rapidly and only lasts on a time scale of seconds~\cite{storozhuk2002post, xue2010post}.\\

As just described, the top-down feedback $t_l$ onto layer $l$ acts as a target for the weights $V^b_{l,l}$ forming the canceling lateral feedback $c_l$. Therefore, the weights are updated during the free phase as follow:
\begin{align}
\dot{V}^b_{l,l} = \eta_V(t_l - c_l) \rho(g_l)^T
\label{eq:mb_dvb}
\end{align}
which minimizes $\|t_l - c_l\|$.

The main weights $W^f_{l+1,l}$ are updated at the end of the weakly-clamped phase through the same local rule as in (MA):
\begin{align}
    \Delta W^f_{l+1,l} = \eta_{W,l} \left(e_{l+1}\cdot \rho'(s_{l+1})\right) \rho(s_l)^T.
    \label{eq:mb_dw}
\end{align}
We used different learning rates $\eta_{W,l}$ for each layer $l$ in the case of (MB).

For a more detailed description of the algorithm, see Algorithm~\ref{algo:mb}.

\paragraph{A good approximation of backpropagation}
These update rules lead to the following theorem, where eq.~\ref{eq:mb_dw} can be again seen as close to backpropagation.

\begin{figure}[ht]
    \centering
  \includegraphics[width=.8\linewidth]{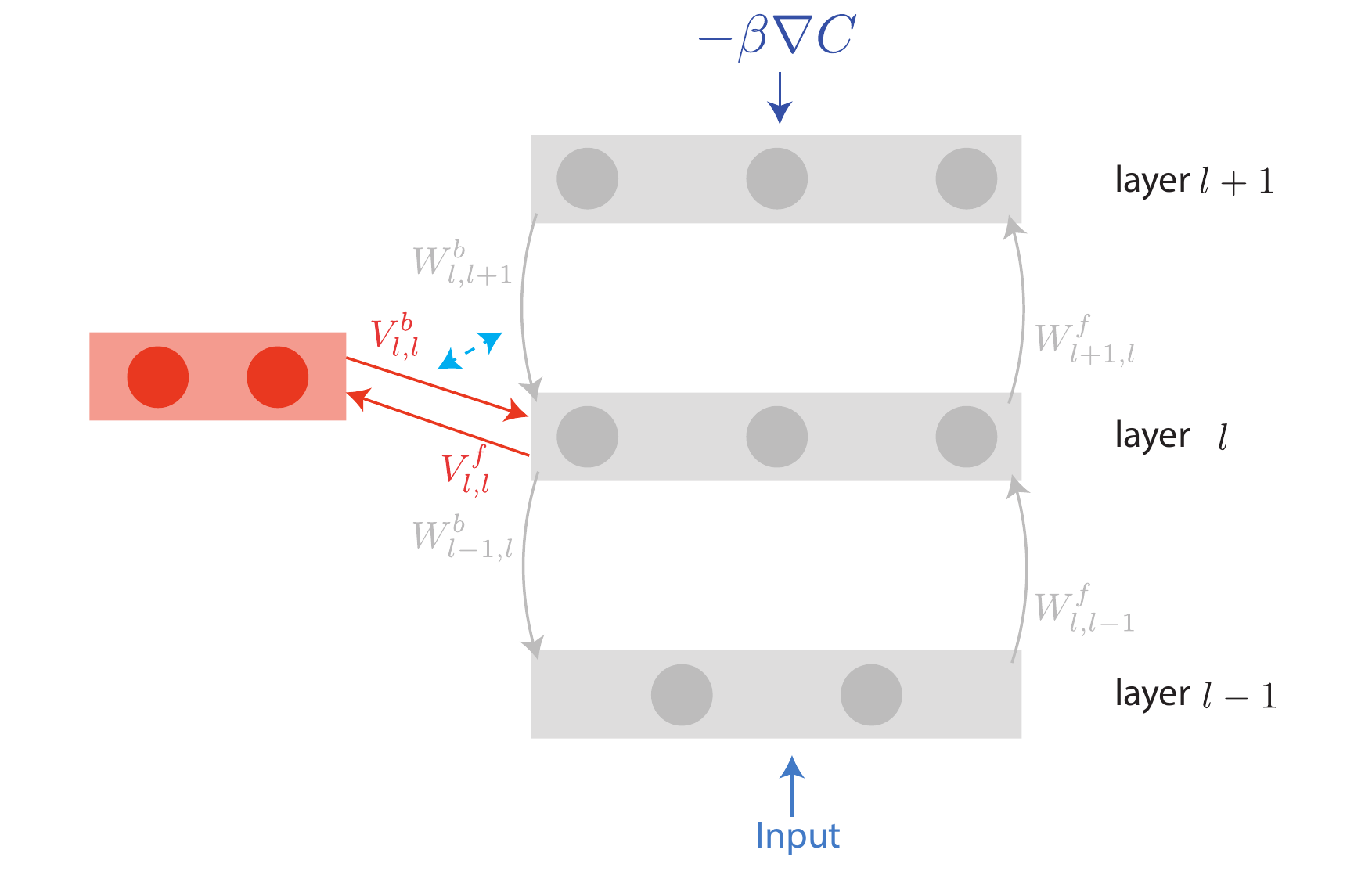}
\caption{Architecture for Model B (MB) network. Pyramidal units (in grey) and ghost units (in orange) are connected through different weights matrices ($W^f$, $W^b$ in grey for PU-PU connectivity and $V^f$, $V^b$ in orange for PU-GU lateral connections). $V^f$ is fixed. $V^b$ thanks to the plasticity rules eq.~\ref{eq:mb_dvb} evolves so that the lateral feedback from the ghost units of layer $l$ replicates the pyramidal units feedback of the following layer $l+1$ (this interaction is represented by the blue dotted line). The nudging (-$\beta \nabla C$) is presented at the output layer ($\beta=0$ in the free phase and $\beta \neq 0$ in the weakly-clamped phase).}
\label{fig:mb_architecture}
\end{figure}

\begin{theorem}
\label{thm:mb_bp}  
We set the backward weights equal to the (adapting) forward weights,  $W^b_{l,l+1}=(W^f_{l+1,l})^T$, and assume that the ghost unit circuit ($V^b$) converged during the free phase (at each different presentation of inputs), $\|t_l - c_l\|\rightarrow 0$ for each hidden layer $l$ (MB). Then for weak output nudging ($\beta$ small, $e_k= -\beta \frac{\partial  C}{\partial  \rho(s_k)}$) errors converge to $e_l \approx -\beta \frac{\diff  \tilde{C}}{\diff  \rho(\tilde{s}_l)}$ for each hidden layer $l$. So, the forward weights ($W^f$) become updated according to the classical backpropagated error gradient. 
\end{theorem}

\begin{proof}
 \emph{Free phase:}
 
  After settling in the free phase ($\dot{s}_l=\dot{g}_l=0$), we have, because $c_l=t_l$:
  \begin{align*}
    s_l = b_l =  W^f_{l,l-1} \rho(s_{l-1})
  \end{align*}
  
  In particular, we have for all units $s_l=\tilde{s}_l$ and, as $c_l=t_l$:
    \begin{equation}
    \label{eq:equMB}
        V^b_{l,l} \rho(g_{l})=W^b_{l,l+1} \rho(\tilde{s}_{l+1})
    \end{equation}  

 \emph{Weakly-clamped phase:}

As in the previous proof, we use induction.

We have clearly that in the output layer $k$:
\begin{equation}
    e_k= -\beta \frac{\partial  C}{\partial  \rho(s_k)} = -\beta \frac{\partial  \tilde{C}}{\partial  \rho(\tilde{s}_k)}
\end{equation}

We will prove the property for the last hidden layer $k-1$, and induction will follow.
Starting from the definition of $e_{k-1}$, we substitute the definition of $t_{k-1}$ and $c_{k-1}$:
  \begin{align}
    e_{k-1} &= t_{k-1} - c_{k-1}  \\
        &= W^b_{k-1,k} \rho(s_k) - V^b_{k-1,k-1} \rho(g_{k-1})
    \label{eq:eiMB}
  \end{align}
  We consider that the layer $k-1$ is still at equilibrium, and that the nudging is small, so $s_{k-1}=\tilde{s}_{k-1}$. And by the same arguments than for (MA):
  \begin{equation}
  \label{eq:mb2}
      s_k=\tilde{s}_k+e_k
  \end{equation}

From eq.~\ref{eq:equMB} (still valid because $g_{k-1}$ has not been impacted by the feedback from above), eq.~\ref{eq:eiMB}, eq.~\ref{eq:mb2} and $e_k$ small, we have: 
  \begin{align}
    e_{k-1} &=  W^b_{k-1,k} \rho(\tilde{s}_k+e_k) - V^b_{k-1,k-1} \rho(g_{k-1})\\
    &\approx W^b_{k-1,k} \rho'(\tilde{s}_k)e_k + W^b_{k-1,k} \rho(\tilde{s}_k) - V^b_{k-1,k-1} \rho(g_{k-1}) \\
    e_{k-1}[i] &\approx \sum_{j \in \mathcal{S}_k} W^b_{k-1,k}[i,j] \rho'(\tilde{s}_k[j])e_k[j]
  \end{align}
 Starting from this point, the proof is the same as for Theorem~\ref{thm:ma_bp}.
 \end{proof}
 
 \subsection{Transpose feedback (TF) versus Feedback-alignment (FA)}

The feedback weights $W^b$ are assumed to be equal to the transpose of the feedforward ones, and are updated as such during the training: $W^b_{l,l+1}=(W^f_{l+1,l})^T$, as in classical backpropagation. We refer to this hypothesis as transpose-feedback (TF). In practice this characteristic could be implemented using an additional reconstruction cost (between consecutive pyramidal layers), which has been shown to encourage symmetry of the weights~\cite{Vincent-JMLR-2010-small}.  This assumption can also be relaxed thanks to~\cite{lillicrap2016random} and the feedback-alignment effect (FA). In this case, feedback weights $W^b$ are fixed and randomly initialized. During learning, the feedforward matrix tends to align with the transpose of the feedback matrix. Both hypotheses (TF) and (FA) were tested here.

\section{Related Work}
\label{sec:related_work}

Backpropagation in the brain has been a very active topic of research for the last few years and various models have been proposed.\\

Constrastive Hebbian learning~\cite{Ackley85,Hinton+McClelland-1988} introduced the idea of learning in two different phases, a free phase where the inputs are presented to the network, followed by a weakly-clamped phase, with a target signal that nudges the output layer towards the right solution. \cite{scellier2017equilibrium} made the parallel between contrastive Hebbian learning and backpropagation with the definition of a framework for energy-based models, Equilibrium Backpropagation. The idea of using two different phases during the training procedure was kept in this work, however, contrarily to the previous studies, we were also able to train the network while allowing synaptic updates during both phases.\\

Segregated dendrites and multicompartment neurons were recently used~\cite{guerguiev2017towards} to implement backpropagation in a biologically plausible manner. This study gave a very interesting explanation of how neurons can store feedforward activity and how feedback connections can carry the backpropagated error without interfering with the feedforward activity. Training can then be performed without dealing with recurrent activity caused by feedback connections, which makes the theory simpler and closer to deep learning methodology. This study achieved great results even if they were using a spiking neural network with update rules being computed using average of the neural potential.\\

A recent study~\cite{sacramento2018dendritic,NIPS2018_8089} introduced the idea of canceling the feedback from the next layer with the inhibitory lateral feedback in order to leave out only the backpropagated error as remaining from the feedback signal. They used recurrent networks of two-compartment neurons. Links can also be drawn with~\cite{Lee:2015:DTP:3120485.3120521,jaderberg2016decoupled,DBLP:journals/corr/abs-1803-01834} where local credit assignment is also performed.\\

As in [17, 18, 19], we effectively consider the dynamics of a single quantity per neuron, the somatic activity. However we do not describe the input currents as coming from different compartments and consider instead a more abstract single-compartment neuron (that, to implement plasticity, is able to represent two quantities, the target rate and its actual rate). This has the advantage of simplifying the terminology of the model as we do not need to introduce dendritic quantities that enter in the representation of the errors. Moreover, this does not induce any scaling of the approximated error by dendritic attenuation factors as in~\cite{sacramento2018dendritic,NIPS2018_8089} (the same scaling can be recovered by multiplying the learning rate by the inverse of the dendritic attenutation factor).  As such, it is possible to approximate the backpropagated gradient without leading to an exponential decay of its magnitude when it is propagated through several layers. We used a reduced system with only the required ingredients in order to obtain a working biologically abstracted analogue of backpropagation. Model A implements in a simple and condensed way the principles from~\cite{sacramento2018dendritic,NIPS2018_8089} where the ghost units network copy the pyramidal one. Diverging from the ideas of Model A, we also postulate a short-term plasticity according to which the local circuit adapts for a single pattern such as in Post-Tetanic Potentiation~\cite{storozhuk2002post, xue2010post} or in the FORCE algorithm~\cite{sussilloabbott}. We develop accordingly Model B where the ghost units dynamically adapt their feedback to replicate in an online manner the feedback coming from the pyramidal units. This single compartment model is sufficient in order to obtain the required credit assignment mechanism and it also simplifies the mathematics to the bare necessities required to obtain the desired results.

\setlength{\abovecaptionskip}{0pt}
\begin{figure}[H]
\centering
\includegraphics[width=0.7\linewidth]{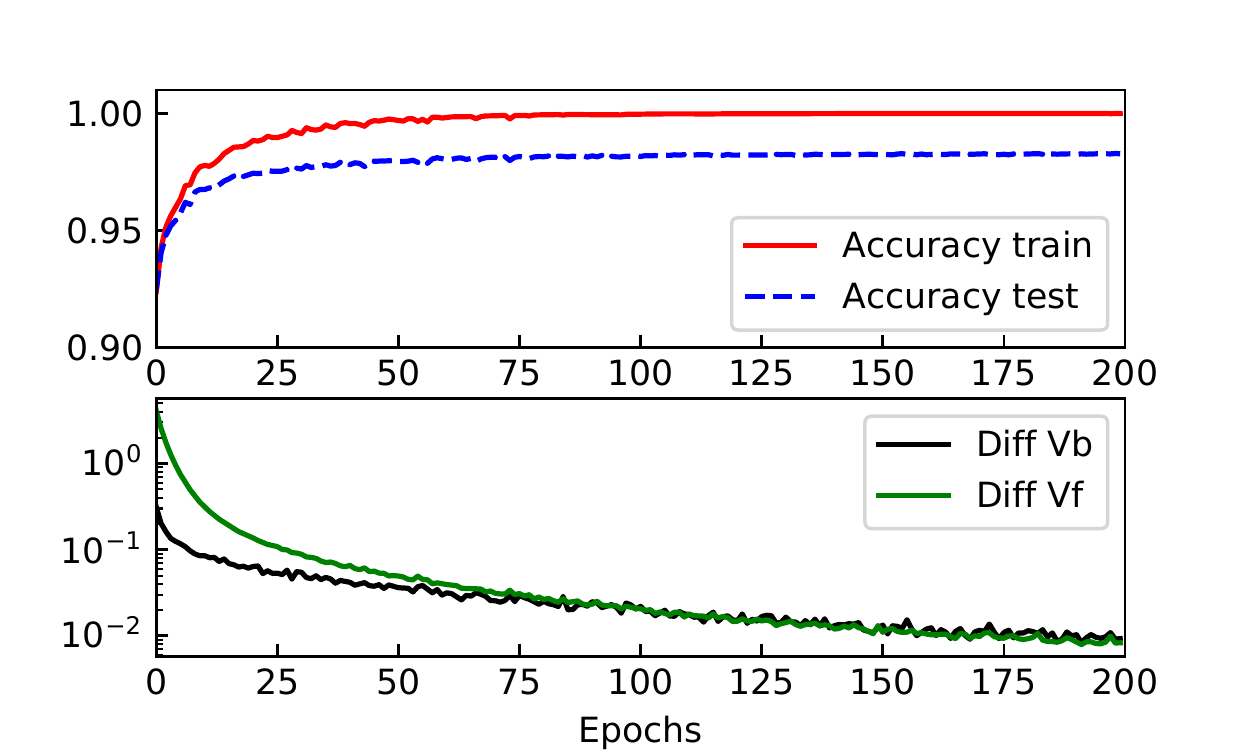}
\caption{Top: Evolution of the train accuracy (dotted) and test accuracy (plain line) on the MNIST classification task for a 500-unit (MA) neural network. Bottom: Evolution across learning of the Frobenius norm between $V^b$ and $W^b$ (Diff Vb) and between $V^f$ and $W^f$ (Diff Vf). Curves are averaged over 5 experiments.}
\label{fig:ma_mean}
\end{figure}

\section{Results}
\label{sec:results}

\subsection{Credit assignment with replicating units}
\label{subsec:credit}

We consider a (784-500-10) network with one hidden layer with MSE (mean-squared error) loss. No preconditioning of the inputs is used. Batch size is equal to 100 for (MA) and 1 for (MB). The activation is sigmoid. We train on the 55000 MNIST training set and test on the 10000 examples of the test set. We initialized the weights randomly with a uniform distribution over $[-\gamma,\gamma]$ (Table~\ref{table:ma_parameters} and~\ref{table:mb_parameters}).

\paragraph{(MA) dynamics}
\label{subsec:ma_dynamics}

In (MA), learning is composed of two phases. During the free phase, only inputs are provided to the network. The weights $V^f$ push the ghost units to mimic their corresponding pyramidal units (eq.~\ref{eq:ma_dvf}) while having $V^b$ learning to minimize the mismatch between the feedback coming from the ghost units and the pyramidal units from the next layer (eq.~\ref{eq:ma_dvb}). This pushes the matrix $V^f$ to reproduce $W^f$ and $V^b$ to copy $W^b$. This can be seen at the bottom of Figure~\ref{fig:ma_mean}, where the Frobenius norms between these matrices during training are reproduced.

This leads to the correct computation of the feedforward path because the feedback terms cancel each other, despite happening in a dynamical way.

During the weakly-clamped phase, the output units are nudged towards the correct values. This shift is backpropagated through the dynamics of the network. This gives rise to an error term at each hidden layer thanks to the mismatch between the feedback coming from the pyramidal units and the corresponding ghost units. $W^f$ evolves in order to minimize this mismatch (eq.~\ref{eq:ma_dw}).

As can be seen at the top of Figure~\ref{fig:ma_mean}, the network learns well to classify MNIST digits (100\% on the train set and 98.27\% on the test set). It generalizes well without any regularization or tricks. 

\paragraph{(MB) dynamics}
\label{subsec:mb_dynamics}

We also studied learning in a neural network following (MB) hypothesis. In the one hidden layer network, 5 ghost units were used which aim to replicate the feedback signal from the 10 output pyramidal units. For the testing of the feedforward network on the train and test MNIST sets, we ran the forward graph without the dynamical part to have quicker simulations.

Different inputs are presented to the network sequentially (no batch in this setting). At each input presentation, the network goes through two different phases (see~Figure~\ref{fig:mb_evolution}).
First, the free phase (in blue) where there is no nudging of the output layer. In particular, the feedback weights adapt their synaptic variables to have the feedback signal from the ghost unit cancel the feedback from the pyramidal cells ($e_l=0$), as can be seen in~Figure~\ref{fig:mb_evolution} (bottom). This cancellation of the local error leads to a correct computation of the feedforward graph of the neural network. Output probabilities of the classification task can then be read from the output pyramidal cells as seen in~Figure~\ref{fig:mb_evolution} (top). 
Then, during the weakly-clamped phase (in green), the output neurons are nudged toward the right solution. This error is then backpropagated through the network by its own dynamics. When the equilibrium is reached, the feedforward weights $W^f$ are updated. Another input is then presented to the network and as such, this process enables learning of the classification task.

\begin{figure}[ht]
\centering
\includegraphics[width=0.7\linewidth]{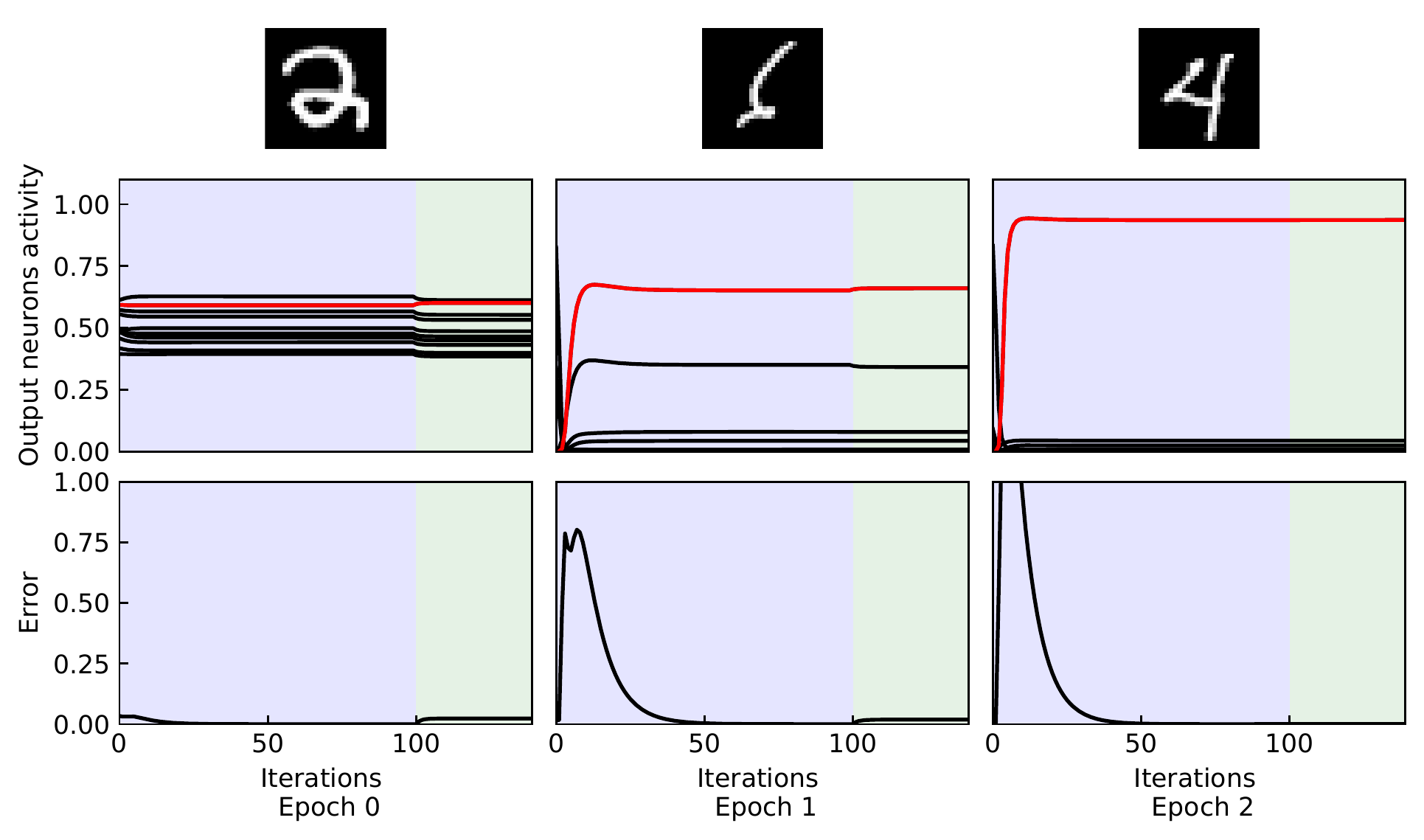}
\caption{Dynamics of a neural network (MB) for 3 different inputs at 3 different times in learning. Level of neuronal activity of the output neurons represented in the top panel (red for the correct class, black for the others). Norm of the difference between the feedback from the ghost units and the pyramidal ones (bottom) as a function of the time step. The free phase is represented in blue and the weakly-clamped phase in green.}
\label{fig:mb_evolution}
\end{figure}

Learning can be studied through the responses of the output neurons at different epochs. At epoch 0 (Figure~\ref{fig:mb_evolution} (left)), the output neurons are mainly wrong. The gradients that are backpropagated have a large amplitude (as can be seen with the jumps in activity after the beginning of the weakly-clamped phase). In particular, we clearly see that the neuron representing the right class (in red) is nudged towards 1 whereas the others are nudged down to 0. After one epoch (Figure~\ref{fig:mb_evolution} (middle)), the output states are moving towards the right solution. However it becomes harder to cancel the error in the free phase, because the weights grow bigger, making the ghost units work harder after a switch between two different inputs. Finally, after two epochs (Figure~\ref{fig:mb_evolution} (right)), the output neurons already start to saturate to 0 or 1. In conclusion, using (MB) dynamics, the neural network is able to quickly learn a classification task.

\begin{figure}[ht]
\centering
\includegraphics[width=.5\linewidth]{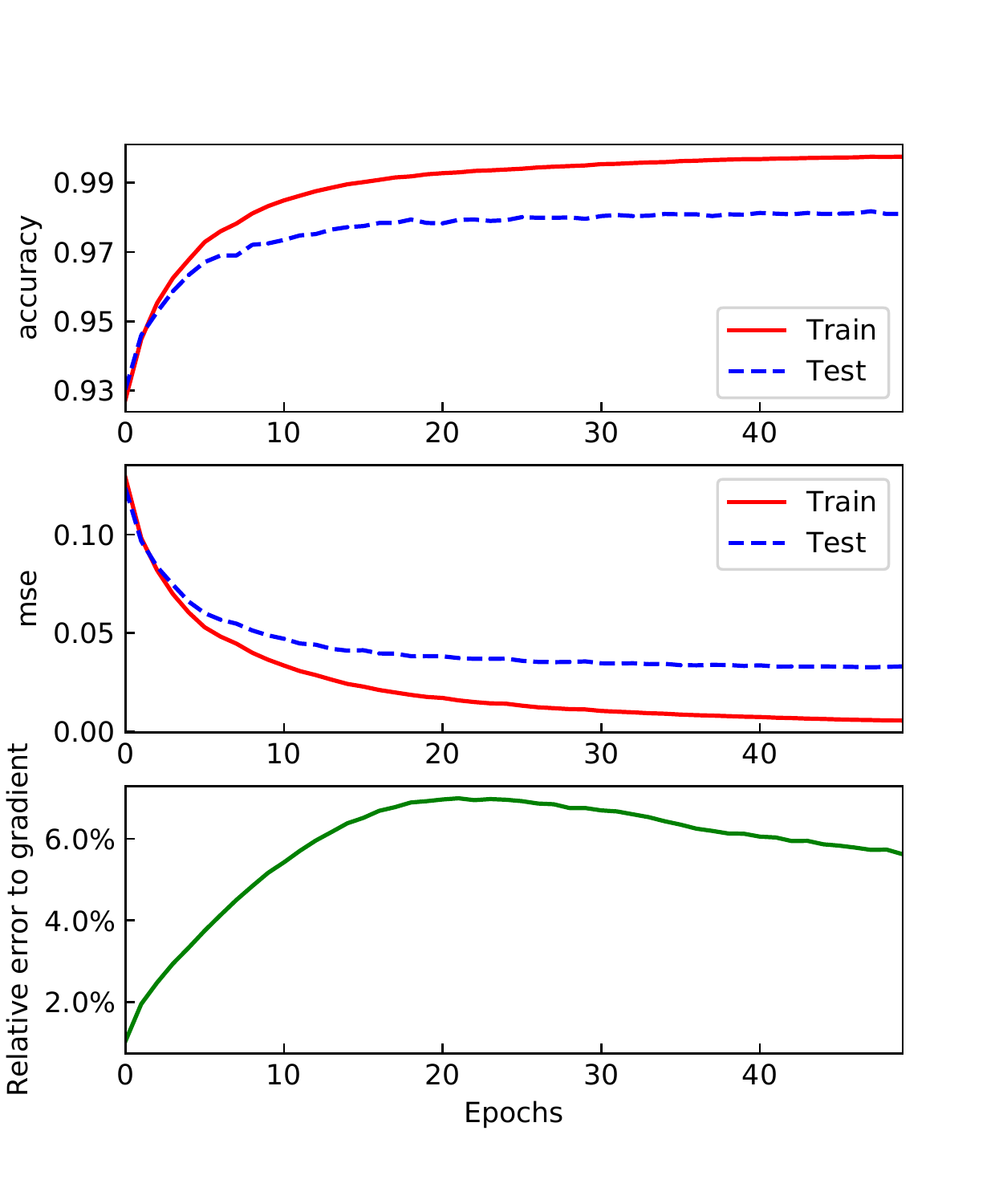}
\caption{Learning in a 500-unit (MB) neural network. Test and training accuracy (top), test and training mean-squared error (middle), relative error to the gradient from classical backpropagation (bottom) as a function of the number of epochs.}
\label{fig:mb_mean}
\end{figure}

On a longer scale of learning (50 epochs), the accuracy and mean-squared error are plotted in~Figure~\ref{fig:mb_mean} for both train and test sets. In particular learning goes well as the training accuracy reach 0.9976 (top). This is correlated with the mean-squared error that goes down and tends to 0 (middle). The gradients computed through the ghost units are almost the same (up to 7\% of relative error to the classical backpropagated gradient, bottom) as the ones that can be computed using the usual backpropagation and the chain-rule. Generalization is also quite good with the accuracy on the test set that reaches 0.981 after only 50 epochs of training, which is quite quick and efficient, considering that none of the usual tricks (Adam, RMSProp, ...) were used in this setting.

As a conclusion, this biologically inspired neural network, following (MB) hypothesis is able to quickly learn in a robust way the classification MNIST benchmark, with locally computed gradients that closely approximate backpropagation.

\subsection{Classification on MNIST}
\label{subsec:classification_results}

\begin{figure}[ht]
\begin{subfigure}{.5\textwidth}
 \includegraphics[width=\textwidth]{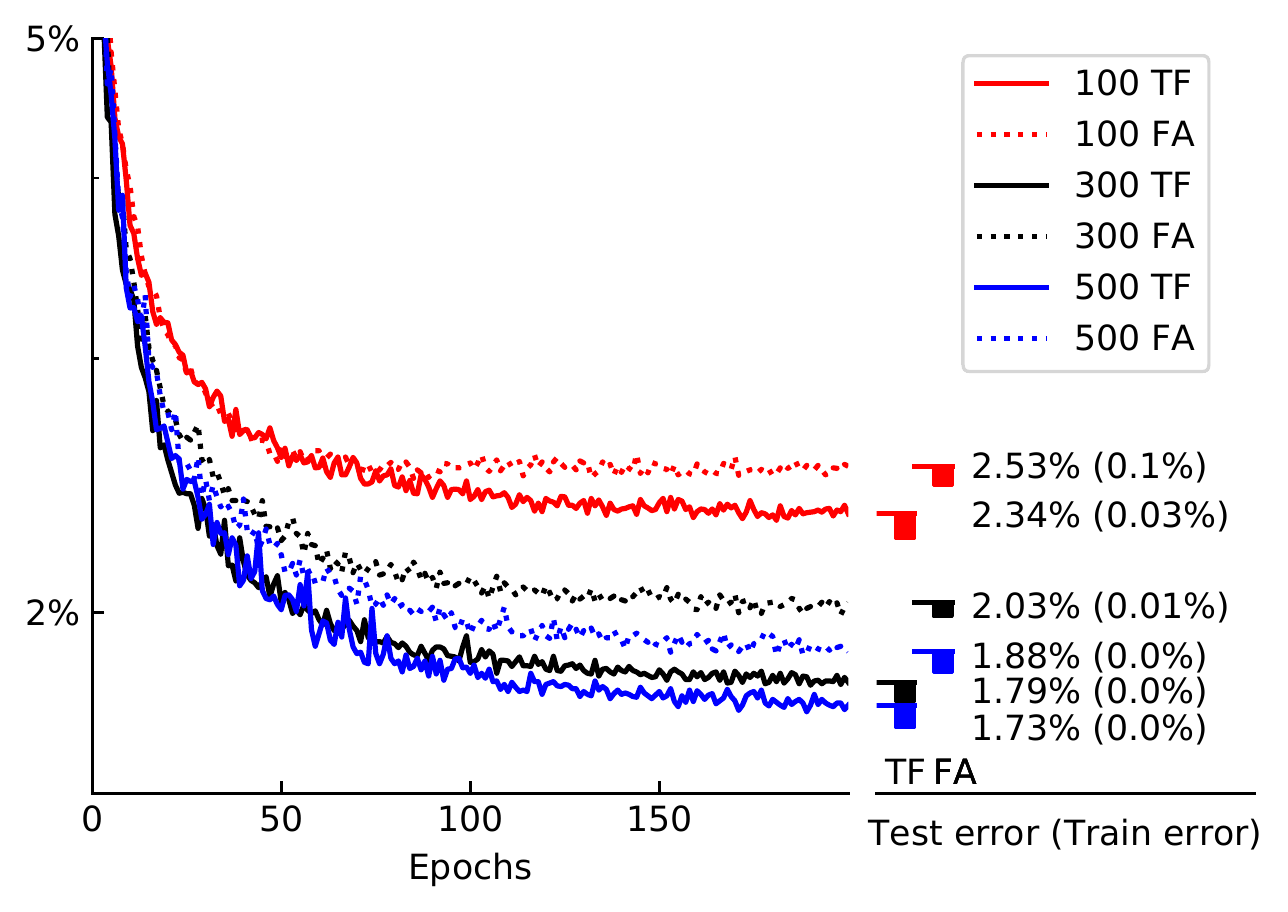}
\caption{Model A (MA)}
\label{fig:ma_results}  
\end{subfigure}
\begin{subfigure}{.5\textwidth}
 \includegraphics[width=\textwidth]{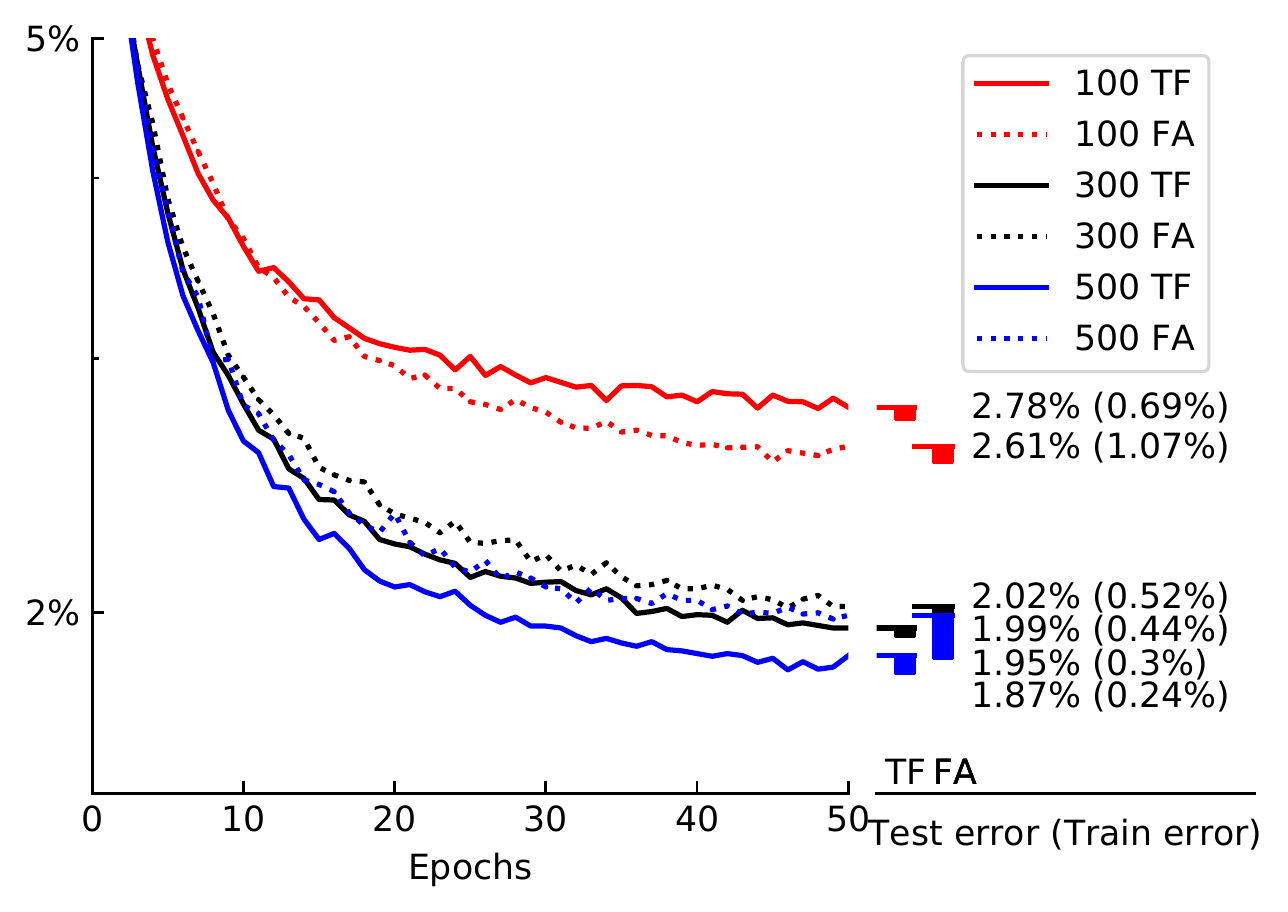}
\caption{Model B (MB)}
\label{fig:mb_results}  
\end{subfigure}
\caption{Classification on MNIST using networks with one hidden layer. Test error on MNIST as a function of the number of epochs for one hidden-layer network with 100 (red), 300 (black) and 500 (blue) neurons. Both TF (plain lines) and FA (dotted lines) are represented. On the right are detailed the mean values of test and train accuracies (between parenthesis). For each network, the standard deviation over 5 experiments is represented by the width of the area below the ticked mean values.}
\end{figure}

We have tested several types of networks on the MNIST dataset~\cite{LeCun+98}, see Table~\ref{table:ma} for the results. We looked at both (MA) and  (MB) models while using either Transpose-feedback (TF) or Feedback-alignment (FA) (see Figure~\ref{fig:ma_results} for (MA) and Figure~\ref{fig:mb_results} for (MB)). We tested different numbers of units per hidden layer and different numbers of hidden layers. Simulations were run on the GPU cluster Cedar, Compute Canada (www.computecanada.ca).

We ran 5 experiments for each model, and represent these results in Table~\ref{table:ma}. The results approach state-of-the-art accuracies for multilayer perceptron, with (MA) performing slightly better than (MB).
For both models, raising the number of neurons resulted in a rise in performance.\\

\paragraph{(MA) performance on the MNIST task}
1-layer and 2-layer (MA) networks compete with the state-of-the-art when using a multilayer perceptron trained with backpropagation. Training with 1-layer (MA) networks is stable and works well when using transpose-feedback and feedback-alignment. Increasing the number of units per layer helps improving performances as shown in~Figure~\ref{fig:ma_results}. Let's note that we were able to use a relatively big $\beta$ during the weakly-clamped phase that speeds up and stabilize learning.

Training with 2-layer (MA) networks is stable and works also well across a large range of hyperparameters when using feedback-alignment. However, when using transpose-feedback training is a bit less stable and requires a more precise hyperparameters search to achieve great performances. Accuracies and hyperparameters are shown in~Table~\ref{table:ma} and~Table~\ref{table:ma_parameters}.

We also realized 1-layer experiments with all synaptic updates occurring at all times during the free and the weakly-clamped phases (MA') (Table~\ref{table:ma_prime}). Theses networks were harder to train but are still able to perform quite well considering all the assumptions that are made (around 3.5\% of test error). These results are comparable to the results from~\cite{guerguiev2017towards,sacramento2018dendritic}.  To get stable behaviors, the updates for $V^b$ were clipped.  Otherwise it sometimes diverged at the beginning of the learning. This hypothesis is biologically plausible, through some saturation mechanisms. 

\paragraph{(MB) performance on the MNIST task}

For (MB) networks, 1-layer network were associated with one layer of $5$ ghosts units (to replicate the feedback from the $10$ pyramidal output units). For 2-layer network, we used two layers of ghost units of respectively $20$ and $5$ neurons. The other parameters are gathered in Table~\ref{table:mb_parameters}.

1-layer (MB) networks perform really well compared to state-of-the-art results, both in the case of transpose-feedback and feedback-alignment. Using networks with more neurons (100 - 300 - 500) also helps reaching higher performance.

For 2-layer (MB) networks with transpose-feedback weights, the networks were harder to train and sometimes unstable. This could be easily explain by the fact that the error that is made when backpropagating the gradient through the biological backpropagation is proportional to the amplitude of the feedback weights. In 2-layer networks with transpose feedback, they grow at the same speed than $W^f$. If they grow too large, they can induce some unwanted errors and make the network unstable. These problems were totally solved when using the feedback-alignment version of (MB), where the amplitude of the feedback weights is fixed. In particular, the higher accuracies for (MB) are reached with the 2-layer version, with feedback-alignment. Adding layers helps the network to perform better and the biological backpropagation is efficient through several layers.

\begin{table}[ht]
\begin{center}
\begin{tabular}{cc|cc|cc} \toprule
    \multicolumn{2}{c}{Architecture} & \multicolumn{2}{c}{TF} & \multicolumn{2}{c}{FA}\\ \toprule
    Model & \#units & Train & Test & Train & Test\\ \toprule
    (MA) & 100    & 99.97 & 97.66 & 99.90 & 97.47 \\
    (MA) & 300    & 100 & 98.21 & 99.99 & 97.97 \\
    (MA) & 500    & 100 & 98.27 & 100 & 98.12 \\
    (MA) & 500/500&99.67& 97.86 & 99.93 & 98.05   \\
    \midrule
    (MA') & 500& 97.77 & 96.57 & - & - \\
    \midrule
     (MB)& 100 & 99.31 & 97.22 & 98.93 & 97.39 \\
     (MB)& 300 & 99.70 & 98.05 & 99.48 & 97.98  \\
     (MB)& 500 & 99.76 & 98.13  & 99.56 & 98.01  \\
     (MB)& 300/300 & 99.84 & 97.95  & 99.78 & 98.05  \\ 
     (MB)& 500/500 & 99.91 & 98.13  & 99.85 & 98.21  \\     \midrule
    \cite{Lillicrap-et-al-nature2016} & 1000 & - & 97.6  & - & 97.9  \\ 
    \cite{nokland2016direct} & 800/800 & - & 98.33  & - & 98.18  \\ 
    \cite{guerguiev2017towards} & 500 & - & 96.4 & - & 95.9 \\
    \cite{guerguiev2017towards} & 500/100 & - & - & - & 96.8  \\
    \cite{sacramento2018dendritic,NIPS2018_8089} & 500/500 & - & 98.04  & - & -  \\
    \bottomrule
\end{tabular}

\caption{Accuracy results on MNIST with (MA) and (MB) in a network with one and two hidden layers.}
\label{table:ma}
\end{center}
\end{table}

\paragraph{Comparison to other works}
We saw that both models (MA) and (MB) were competing with state-of-the-art multilayer perceptrons on MNIST, as in~\cite{Lillicrap-et-al-nature2016,nokland2016direct} where classical backpropagation is used.
As we can see in~Table~\ref{table:ma}, we get higher accuracies than other previous biologically-plausible models~\cite{guerguiev2017towards,sacramento2018dendritic,NIPS2018_8089} in our setting with different updates rules during the two phases. Even when updating all weights during both phases (MA'), we get similar results than a model with segregated dendrites compartment ~\cite{guerguiev2017towards}, where updates are done in two different phases, but with a spiking neural network. We however cannot reach an accuracy as high as in~\cite{sacramento2018dendritic,NIPS2018_8089} when considering updates in both phases as they do, this may be due to the fact they use different compartments. In conclusion, both models presented in this work were able to reach accuracies comparable to backpropagation, with a simple biologically-plausible setting.

\section{Conclusion}
\label{sec:conclusion}

Deep learning has been the focus of intense studies in the past decade and has become more and more efficient in solving diverse and complex tasks, that range from pattern recognition, to image generation, NLP (natural language processing) and many others. Many ideas that had some success in deep learning originally come from neuroscience, and making link between both worlds has been the focus of many studies recently~\cite{Marblestone-et-al-2016}. In particular, different mechanisms have been developed to update the parameters of an artificial neural network. Backpropagation~\cite{almeida1987learning} has been the canonical and most used way of training a network. Many frameworks have been developed to enable efficient gradient computations thanks to the use of the chain rule~\cite{2016arXiv160502688short,tensorflow2015-whitepaper}. However backpropagation as it is commonly used has many properties that are not biologically plausible~\cite{bengio2015towards, neftci2017event}.\\

We developed in this work a class of neural networks models, that enables learning as would backpropagation but with local learning rules. To the feedforward neural network of pyramidal units (which represents a classical multilayer perceptron) we add a second network of inhibitory interneurons, which we denote as ghost units. Connections between all layers from both networks exist making the whole network a complex recurrent one. The dynamics of the models are separated into two phases as in~\cite{scellier2017equilibrium}. During the free phase, the ghost units network learns to perfectly replicate the feedback from the pyramidal units, and therefore cancel any feedback coming from the upper layers. At equilibrium, the network propagates the correct feedforward graph and can perform a classification task. During the weakly-clamped phase, the output pyramidal units are nudged in order to reduce the output cost function. The error is backpropagated through the layers thanks to the recurrent dynamic itself. We considered two different models, the first one (MA) where each pyramidal unit has an associated ghost unit in the previous layer. The ghost unit network learns throughout the training to replicated the pyramidal one. In the second model (MB), we consider fewer ghost units that dynamically learn to replicate the feedback from the pyramidal units at each different pattern presentation. We prove for both models that, under some hypotheses, the locally defined learning rules approximate classical backpropagation, with condensed notations and proofs. Moreover, we tested both model on the MNIST classification tasks, with different architectures (number of neurons, number of layers) and proved that these networks were able to accomplish such tasks as well as with backpropagation. We also made links with feedback-alignment~\cite{Lillicrap-et-al-nature2016}. Finally, we were able to loosen some of the assumptions made in our models such as updating all the synaptic weights at all time, while still having a trained network that was able to perform correctly in the pattern recognition task. This single compartment model only expresses the bare properties needed to have a credit assignment mechanism, and as such, with simple notations and proof, straightforwardly highlights a possible implementation of credit assignment in the brain.\\

The class of models presented in this work was built upon some properties of the mammalian cortex. In particular, the different layers of the pyramidal neural network represent the integration of the sensory stimulus (from other brains areas, thalamus for example) through several layers of cortical neurons (through higher order regions of the brain). As in~\cite{sacramento2018dendritic,NIPS2018_8089}, we also consider a population of inhibitory interneurons that would be coherent with the neurophysiological properties and role of SST interneurons. In particular, their role would be to cancel top-down feedback from the pyramidal neurons. Some recent experiments~\cite{LEINWEBER20171204} showed that pyramidal neurons project back through top-down projections to the interneurons of the previous layer, which would be coherent with these interneurons replicating the activity from upper layers. Other works such as~\cite{ENIKOLOPOV2018135} showed that synaptic plasticity could generate a negative image of the input in the electric fish and they illustrate the importance of this kind of signal for improvements in neural coding and detection of perturbations. The models presented in this work are able to backpropagate the output layer error, through the different layers without needing gradients computations. The network is able to learn thanks to local learning rules, which have been proven to have some biological relevance and links to STDP (spike-timing dependent plasticity)~\cite{bengio2017stdp,Feldman-2012}. We use single-compartment leaky-integrate neurons, that can be seen as a really simple approximation of biological neurons. 

However some properties of the presented network can hardly be seen as biologically plausible. Firstly, in (MA), we considered that there was a 1-1 correspondence between the pyramidal cells and the ghost units. This is hardly true in biological neuronal network, however we could consider that each pyramidal cell models the activity of several pyramidal neurons and as such, forget this hypothesis. We also developed (MB), where the number of ghost units can be set to an arbitrary (smaller) number and forget 1-1 correspondence. For this we suppose that the network of ghost units is able to dynamically replicate the feedback from the pyramidal cells through some rapid plastic mechanisms (as can be seen with Post-Tetanic Potentiation~\cite{storozhuk2002post, xue2010post}), which may not be biologically implemented directly without any neurotransmitter modulation. However, as we present two models that rely on opposite hypotheses, it would be possible to compose a large class of models sharing properties from both of these extreme cases. We could then have, at the same time, an adaptive process, that learns on long time scale how to cancel feedback but that can also adapt to variations in the input, making it closer to biological observations. One possible way to implement such a neural network would be to make the ghost units of (MA) replicate a linear combination of the pyramidal units activity and as such use fewer ghost units. Training 2-layer networks with transposed feedback weights was in some case unstable due to the fact that the transmission of the backpropagated error was proportional to the amplitude of the feedback weights, and therefore of the forward ones. Using feedback-alignment solves this issue, as the scale of the feedback weights was fixed. However we can imagine that this problem could also be overcome in the brain thanks to normalization mechanisms such as weight decay or weights regularization, that could be implemented by biologically plausible mechanisms.

We could also implement integrate-and-fire neurons and make the link with spiking deep networks~\cite{NIPS2011_4383} in order to go further towards biological networks. We focused on pattern recognition in this work, but ghost units could also help to make biologically plausible networks for other tasks of deep learning such as Generative Adversarial Networks~\cite{Goodfellow-et-al-NIPS2014-small} or Deep Reinforcement Learning, by adding some other features such as the influence of neurotransmitters on local learning rules.

This work is a step towards apprehending the mechanics of learning and memory in the brain, but at the same time, it raises interesting perspectives for implementation of deep networks in neuromorphic hardware.

\bibliographystyle{unsrt}
\bibliography{main}

\begin{table}[ht]
\begin{center}
\begin{tabular}{l|ccc|cc}
\toprule
    Hidden Layers            & 1 & 1 & 1 & 2  \\
    Pyramidal Units per layer (784, X, 10)  & 100 & 300 & 500 & 500, 500 \\
    Ghost Units per layer (-, X, -)    & 10 & 10 & 10 & 500, 10 \\   \midrule
    $\diff t$                & 0.001 & 0.001 & 0.001 & 0.001 \\
    $\tau$                & 0.01 & 0.01 & 0.01 & 0.01 \\
    $\gamma$  & 0.2 & 0.2 & 0.2 & 0.1 \\
    Free Steps         & 200 & 200 & 200 & 200 \\
    Weakly-clamped Steps            & 200 & 200 & 200 & 200  \\
    Nudge size ($\beta$)     & 10 & 10 & 10 & 10 \\
    Batch Size               & 100 & 100 & 100 & 100  \\
    Epochs                   & 200 & 200 & 200 & 300 \\ \midrule
    Learning Rate $W$ ($\eta_W$) TF     & 0.1, 0.1 & 0.1, 0.1 & 0.1, 0.1 & 0.01, 0.01, 0.01  \\
    Learning Rate $V$($\eta_V$) TF     & 0.05 & 0.05 & 0.05 & 0.01, 0.01 \\
    Accuracy Train TF        & 99.97 & 100 & 100     & 99.67\\
    Accuracy Test TF         & 97.66 & 98.21 & 98.27 & 97.86 \\\midrule
    Learning Rate $W$ ($\eta_W$) FA     & 0.1, 0.1 & 0.1, 0.1 & 0.1, 0.1 & 0.05, 0.05, 0.05  \\
    Learning Rate $V$ ($\eta_V$) FA     & 0.05 & 0.05 & 0.05 & 0.02, 0.02 \\
    Accuracy Train FA        & 99.90 & 99.99 & 100 & 99.93\\
    Accuracy Test FA         & 97.47 & 97.97 & 98.12 & 98.05 \\\bottomrule
\end{tabular}
\caption{Parameters for MNIST classification with (MA)}
\label{table:ma_parameters}
\end{center}
\end{table}
    
\begin{table}[ht]
\begin{center}
\begin{tabular}{l|c}
\toprule
    Hidden Layers            & 1   \\
    Pyramidal Units per layer (784, X, 10)  & 500 \\
    Ghost Units per layer (-, X, -)    & 10 \\   \midrule
    $\diff t$                & 0.001  \\
    $\tau$                & 0.01  \\
    $\gamma$  & 0.2 \\
    Free Steps         & 200  \\
    Weakly-clamped  Steps            & 200  \\
    Nudge size ($\beta$)     & 0.2  \\
    Batch Size               & 100  \\
    Epochs                   & 100 \\ \midrule
    Learning Rate $W$ ($\eta_W$) TF     & 0.2, 0.2  \\
    Learning Rate $V$($\eta_V$) TF     & 0.03  \\
    Clipping $\Delta V^b$         &2\\
    Accuracy Train TF        & 97.77 \\
    Accuracy Test TF         & 96.57 \\\bottomrule
\end{tabular}
\caption{Parameters for MNIST classification with (MA')}
\label{table:ma_prime}
\end{center}
\end{table}

\begin{table}[ht]
\begin{center}
\begin{tabular}{l|ccc|cc}
\toprule
    Hidden Layers            & 1 & 1 & 1 & 2 & 2  \\
    Pyramidal Units per layer (784, X, 10)  & 100 & 300 & 500 & 300, 300 & 500, 500 \\
    Ghost Units per layer (-, X, -)        & 5 & 5 & 5 & 20, 5 & 20, 5 \\   \midrule
    $\diff t$                & 0.005 & 0.005 & 0.005 & 0.005 & 0.005  \\
    $\tau$                & 0.01 & 0.01 & 0.01 & 0.01 & 0.01  \\
    $\gamma$ & 0.05 & 0.05 & 0.05 & 0.05 & 0.05 \\
    Weakly-clamped Steps            & 40 & 40 & 40 & 40 & 40  \\
    Nudge size ($\beta$)     & 0.1 & 0.1 & 0.1 & 0.1 & 0.1 \\
    Learning Rate $V$ ($\eta_V$)        & 20 & 20 & 20 & 20 & 20  \\
    Batch Size               & 1 & 1 & 1 & 1 & 1  \\
    Epochs                   & 50 & 50 & 50 & 50 & 50 \\ \midrule
    Free Steps TF           & 100 & 100 & 100 & 140 & 140 \\
    Learning Rate $W$ ($\eta_W$) TF     & 4, 0.04 & 4, 0.04 & 4, 0.04 & 40.0, 0.1, 0.005 & 40.0, 0.1, 0.01  \\
    Accuracy Train TF        & 99.31 & 99.70 & 99.76 & 99.84 & 99.91\\
    Accuracy Test TF         & 97.22 & 98.05 & 98.13 & 97.95 & 98.13 \\\midrule
    Free Steps FA           & 100 & 100 & 100 & 100 & 100 \\
    Learning Rate $W$ ($\eta_W$) FA     & 8, 0.2 & 8, 0.2 & 8, 0.2 & 200.0, 2.0, 0.02 & 200.0, 1.0, 0.01  \\
    Accuracy Train FA        & 98.93 & 99.48 & 99.56 & 99.78 & 99.85\\
    Accuracy Test FA         & 97.39 & 97.98 & 98.01 & 98.05 & 98.21\\\bottomrule
\end{tabular}
\caption{Parameters for MNIST classification with (MB)}
\label{table:mb_parameters}
\end{center}
\end{table}

\end{document}